\documentclass[12pt,reqno]{amsart}

\usepackage{amsmath,amsfonts,amsthm,amssymb,amsxtra, amsbsy}
\usepackage{color}
\usepackage{hyperref}

\newtheorem{thm}{Theorem}[section]

\newtheorem{lem}[thm]{Lemma}

\theoremstyle{definition}

\newtheorem{rem}[thm]{Remark}

\numberwithin{equation}{section}  

\begin{document}

\title{Absence of a Ground State for Bosonic Coulomb Systems with Critical Charge}

\author{Yukimi Goto}
\address{Graduate School of Mathematical Sciences,
The University of Tokyo, 3-8-1 Komaba Meguro-ku Tokyo 153-8914, Japan}
\email{yukimi@ms.u-tokyo.ac.jp}


\begin{abstract}
We consider bosonic Coulomb systems with $N$-particles and $K$ static nuclei.
Let $E_N^Z$ denote the ground state energy of a bosonic molecule of the total nuclear charge $Z$. 
We prove that the system has no normalizable ground state when $E_{N}^{N-1} = E_{N-1}^{N-1}$.
\end{abstract}

\maketitle

\section{Introduction}
We consider a molecule consisting of
$N$-particles and $K$ fixed nuclei with positive charges $z_1, \dots, z_K > 0$ 
located at distinct positions $R_1, \dots, R_K \in \mathbb{R}^{3}$.
Let $Z = \sum_{i=1}^K z_i$ be the total nuclear charge. 
This system is described by the Hamiltonian
\[
H_N^Z = \sum_{i=1}^N \left(-\Delta_i - 
\sum_{j=1}^K z_j|x_i - R_j|^{-1}\right) + \sum_{1 \le i < j \le N} |x_i - x_j|^{-1}
+ \sum_{1 \le i < j \le K}z_iz_j|R_i - R_j|^{-1}
\]
acting on $L^2(\mathbb{R}^{3N})$. 
Here  $x_i \in \mathbb{R}^3$ and $\Delta_i$ are, respectively, 
the particle coordinates
and the three-dimensional Laplacian with respect to the coordinate $x_i$.
It is well-known that $H_N^Z$ is a self-adjoint operator with 
the domain $H^2(\mathbb{R}^{3N})$ and bounded from below.
The ground state energy is defined by
\[
E_N^Z = \inf \mathrm{spec} \, 
H_N^Z = 
\inf \left\{ \left\langle\psi, H_N^Z\psi \right\rangle \, :  \, \psi \in H^1(\mathbb{R}^{3N}),
 \, \|\psi\|_2 = 1 \right\}
\]
and if it is an eigenvalue, the corresponding eigenfunction is called the ground state.
We impose no symmetry requirements on the particles for our study.
That actually the unrestricted ground state wave function is the same as the bosonic (totally symmetric) one is a conclusion of ~\cite[Section 3.2.4]{LS}.
It is always the case that $E_N^Z \le E_{N-1}^Z$
 (see, e.g.,~\cite{Lewin20113535}).
 
According to the HVZ theorem~(\cite{Lewin20113535}), 
$E_N^Z < E_{N-1}^Z$ implies that there exists a ground state eigenfunction of $H_N^Z$.
Zhislin~\cite{Zhislin} proved that  $E_N^Z < E_{N-1}^Z$ for $Z > N-1$ 
and it is also known that  the system is not bound if for a given $N$ the total nuclear charge becomes sufficiently small.
More precisely, the system has no ground state if  $N \ge 2Z + K$~(\cite{PhysRevA.29.3018}).
This implies instability of the di-anion $\mathrm{H}^{2-}$.
In the usual fermionic case, experimental~\cite{doi:10.1063/1.556047} and numerical~\cite{0953-4075-31-10-001, QUA:QUA18} evidence suggests that there are no stable atomic di-anions $\mathrm{X}^{2-}$,
i.e., fermionic atoms are not bound if $N > Z + 1$  (see,~\cite{LS} for further information).

In the critical case $Z=N-1$, it might happen either $E_N^{N-1} < E_{N-1}^{N-1}$ or $E_N^{N-1} =E_{N-1}^{N-1}$.
The latter case leads to absence of anions (e.g., presumably, $\mathrm{He}^{-}$, $\mathrm{Be}^{-}$, etc)
 or, otherwise,  existence of bound states having zero binding energy as well~\cite{0305-4470-16-6-007, doi:10.1142/S0129055X13500219, Gridnev}.
In this paper, we prove that the ground states of bosonic molecules are delocalized when $E_N^{N-1} =E_{N-1}^{N-1}$.

\begin{thm}
\label{main}
Suppose   $E_N^{N-1} =E_{N-1}^{N-1}$. 
Then there cannot be a normalizable ground state of $H_N^{N-1}$ in $L^2(\mathbb{R}^{3N})$.
\end{thm}

Thus, bosonic anions $\mathrm{X}^{-}$ fail to be stable in that case.
In nature, fermionic $\mathrm{He}^{-}$ anion ($N=3$, $Z=2$) is unstable as also numerical computations show \cite{0953-4075-41-2-025002} (but a virtual state can be expected as indicated by~\cite{1402-4896-58-1-004}).
On the other hand, Hogreve~\cite{doi:10.1007/s10955-011-0265-0} proved bosonic $\mathrm{He}^{-}$ can exist as a stable anion.

For the  atom with $N=2$ electrons, the anion $\mathrm{H}^{-}$ ($Z=1$) has a ground state
~(\cite{PhysRevA.10.1109, PhysRevLett.112.173001}).
Moreover, for bosonic atoms it is known that $E_N^Z < E_{N-1}^Z$ for all
$N \le N_c (Z) = 1.21Z + o(Z)$
as $Z \to \infty$, where $N_c (Z)$ is the maximum number of particles
 that can be bound to a nucleus of charge $Z$ (\cite{PhysRevLett.50.1771, Bach1991}).
 In particular, $E_N^{N-1} < E_{N-1}^{N-1}$ for $N$ sufficiently large.
As yet it is unknown if the situation $E_N^{N-1} =E_{N-1}^{N-1}$ is true for some $N$.

\begin{rem}
Let $Z_c > 0$ be a critical value such that for $Z > Z_c$ one has
$E_N^{Z} < E_{N-1}^{Z}$, and  $E_N^{Z_c} = E_{N-1}^{Z_c}$.
Clearly, $Z_c \le N-1$ by Zhislin's theorem.
Our theorem corresponds to the case $Z_c = N-1$.
In the atomic situation $K=1$, 
it was shown in~\cite{doi:10.1142/S0129055X13500219, Gridnev}, 
that $H_N^{Z_c}$ has a ground state
if $Z_c < N-1$ and $E_{N-1}^{Z_c} < E_{N-2}^{Z_c}$ (in~\cite{Gridnev},
if $Z_c \in (N-2, N-1)$).
Furthermore, these results are  also valid for the
fermionic case.
But our proof works only for the unconstrained (bosonic) case
because the positivity of the ground state is needed.
\end{rem}

\begin{rem}
It has been shown by Gridnev~\cite{Gridnev} that for certain combinations of charges the ground state in the three-particle Coulomb system is delocalized when the energy of the system equals the bottom of the continuous spectrum. 
\end{rem}

\begin{rem}
The symmetry properties of the wave functions may have an important role in the existence and absence of the bound states at threshold.
For instance, we consider the two-particle Hamiltonian $H(Z) = - \Delta_x + ZV(|x|)$, where $V \in C_c^{\infty}(\mathbb{R}^3)$ and $V \le 0$, $Z > 0$ (the center of mass motion is removed).
There would exist $Z_c$ such that $\inf \mathrm{spec} \, H(Z) \nearrow 0$ as $Z \searrow Z_c$. If the particles are bosons or boltzons (without any symmetry restrictions) then we can repeat the superharmonic argument in the proof of Theorem \ref{main}, and thus the ground state of $H(Z_c)$ is not localized.
If the particles are spatially anti-symmetric ($\psi(x) = -\psi(-x)$) then the ground state is bound and localized (\cite[Theorem 2.4]{AnnPhys.130.251}).

In the three particle problem without particle statistics and with essential spectrum starting at zero the ground states are delocalized only in special cases. We refer to~\cite{Gridnev2}.
\end{rem}

\section{Proof of Theorem \ref{main}}
Although, as already explained in the introduction, the unconstrained ground state is the same as the bosonic one, the particle symmetry will be never used in any of the proofs.
 Our method is similar to  that of M. and T. Hoffman-Ostenhoff~\cite{0305-4470-17-17-009}.
 They showed that the Hamiltonian of the two-particle atom
 $H_2^{Z_c}$ has no ground state in the triplet S-sector at the threshold $Z_c = 1$. 
The triplet S-sector means that
the admissible functions  in $L^2(\mathbb{R}^6)$ are restricted to
$\psi(x_1, x_2) = - \psi(x_2, x_1)$ and there exists a $\phi \in L^2(\mathbb{R}_+\times \mathbb{R}_+\times \mathbb{R})$ with
$\psi(x_1, x_2) = \phi(|x_1|, |x_2|, x_1 \cdot x_2)$.
 
  Our proof is simpler than the one in~\cite{0305-4470-17-17-009}, 
  and we treat the molecule of arbitrary $N$-particles,
  but without the anti-symmetry assumption (the Pauli exclusion principle).

Our proof proceeds by constructing a contradiction.
We suppose that there is a $\psi \in H^1(\mathbb{R}^{3N})$ 
so that $H_N^Z \psi = E_N^Z \psi$ and $\psi \neq 0$ for $Z = Z_c = N - 1$.
 The ground state eigenfunction $\psi$ is automatically in $H^2(\mathbb{R}^{3N})$,
 and we may assume that
 $\psi \ge 0$ (\cite[Theorem 7.8]{LiLo}).
Then, by the results of Aizenman and Simon~\cite[Theorem 3.10 \& 1.5]{CPA3160350206} ,
  $\psi$ is strictly positive and continuous.

For any function $g : \mathbb{R}^{3N} \to \mathbb{R}$
let
 \[
 [g]_{x_N}(x_1, \dots, x_{N-1}, r_N)
 =  \int_{\mathbb{S}^2} g(x_1, \dots, x_N)\frac{d\omega_N}{4\pi}
 \]
 be the spherical average of $g$ with respect to the variable $x_N = (r_N, \omega_N)$,
where $d \omega_N$ is the spherical measure on $\mathbb{S}^{2}$,
 which is the unit sphere in $\mathbb{R}^3$.
We note that $\left\| [g]_{x_N} \right\|_2 \le \|g\|_2$ for any $g$ by the Cauchy-Schwarz inequality.

The next lemma is a slight modification of~\cite[Lemma 7.17]{RevModPhys.53.603}.
 
 \begin{lem}
 \label{calculus}
Let $\psi$ be a strictly positive, continuous function in $H^2(\mathbb{R}^{3N})$
and set
$f(x_1, \dots, x_N) = \exp \left([\ln \psi]_{x_N} (x_1, \dots, x_{N-1}, r_N)\right)$.
Then $f > 0$, $f \in C(\mathbb{R}^{3N})$, and
\[
\left[ \frac{\Delta_i \psi}{\psi}\right]_{x_N}f
\ge \Delta_i f, \quad 
\]
$i = 1, \dots, N$, in the sense of distributions.
Moreover, the weak derivatives obey
$\nabla_i f \in L^2_{\mathrm{loc}}(\mathbb{R}^{3N})$ and
$\Delta_i f \in L^1_{\mathrm{loc}}(\mathbb{R}^{3N})$, for $i=1, \dots, N$.

 \end{lem}
 
 \begin{proof}
 Note that $f \le [\psi]_{x_N} \in L^2(\mathbb{R}^{3N})$ 
 from Jensen's inequality~(\cite[Theorem 2.2]{LiLo}). 
Let $\rho_{\varepsilon}$ be a mollifier in $\mathbb{R}^{3N}$,
namely, 
$\rho_{\varepsilon}(x) 
= \varepsilon^{-{3N}} \rho (\varepsilon^{-1} x$), for $\varepsilon > 0$, where $\rho$ is a function in
  $ C_c^{\infty}(\mathbb{R}^{3N})$
 (infinitely differentiable functions with compact support)
 satisfying $\rho \ge 0$, $\int \rho = 1 $, $\rho (x) = 0$ for $|x| \ge 1$.
 Then  $\psi_{\varepsilon} = \rho_{\varepsilon} \ast \psi$ is strictly positive,
smooth and  satisfies
 that $\psi_{\varepsilon} \to \psi$ in $H^2(\mathbb{R}^{3N})$ (\cite[Thoerem 2.16 \& Lemma 6.8]{LiLo})
and $\psi_{\varepsilon} \to \psi$ uniformly on any compact set as $\varepsilon \to 0$,
where the symbol $\ast$ denotes the convolution of two functions.
We introduce a strictly positive, continuous function defined by  $f_{\varepsilon} = \exp \left( \left[ \ln \psi_{\varepsilon} \right]_{x_N} \right)$.
It is easy to see that $f_{\varepsilon}$ also converges to $f$ uniformly on any compact set,
since, by the mean value theorem of calculus, 
 $|f_{\varepsilon} - f| \le \mathrm{const.} \left[ \left|\psi_{\varepsilon} - \psi \right| \right]_{x_N}$
 on any compact set.
By direct computation, we see that
 \begin{equation}
 \label{0}
 \Delta_i f_{\varepsilon} =
 \left( \Delta_i [ \ln \psi_{\varepsilon}]_{x_N}\right) f_{\varepsilon} +
  \left| \nabla_i [\ln \psi_{\varepsilon}]_{x_N}\right|^2 f_{\varepsilon}
 \end{equation}
 and
 $\Delta_i [\ln \psi_{\varepsilon}]_{x_N} = [\Delta_i \ln \psi_{\varepsilon}]_{x_N} $ for $i = 1,\dots, N$, since $[\Delta_N g]_{x_N} =
 [\left(\partial^2/{\partial r_N^2} + 2/{r_{N}}\partial/{\partial r_N}\right)g]_{x_N}$ 
for $g \in C^2(\mathbb{R}^3)$.
 
 As $\varepsilon \to 0$, $\Delta_i f_{\varepsilon}$ converges to
 \begin{equation}
\label{1}
 F_i =
 \left( \left[ \frac{\Delta_i \psi}{\psi} \right]_{x_N}
 - \left[ \left|\frac{\nabla_i \psi}{\psi} \right|^2 \right]_{x_N}
 + \left| \nabla_i \left[ \ln \psi\right]_{x_N} \right|^2\right)f,
 \end{equation}
$i=1, \dots, N$, in $L^1(K)$ for any compact set $K \subset \mathbb{R}^{3N}$.
 In fact, for any compact set $K$, we may choose a constant $C > 0$ such that $\psi$, $\psi_{\varepsilon} \ge C$ on $K$ by the positivity and continuity of $\psi$. Since $|a_{\varepsilon}b_{\varepsilon} - ab| =
|(a_{\varepsilon} - a)b_{\varepsilon} + a(b_{\varepsilon}-b)|$,
a simple calculation shows that
 \begin{align*}
  \left| \left[ \frac{\Delta_i \psi_{\varepsilon}}{\psi_{\varepsilon}} \right]_{x_N}f_{\varepsilon}
 -  \left[ \frac{\Delta_i \psi}{\psi} \right]_{x_N}f \right|
&=   \left| \left[ \frac{\psi \Delta_i (\psi_{\varepsilon} - \psi) + (\psi - \psi_{\varepsilon})\Delta_i\psi}{\psi_{\varepsilon} \psi} \right]_{x_N}f_{\varepsilon} 
 +  \left[ \frac{\Delta_i \psi}{\psi} \right]_{x_N}(f_{\varepsilon} - f) \right| \\
&\le \mathrm{const.} 
\left(
  \left[ 
    \left|
      \Delta_i 
       \left( \psi_{\varepsilon} -  \psi 
       \right)
    \right| 
  +  \left|
        \psi - \psi_{\varepsilon}
      \right|
   \left| \Delta_i \psi 
   \right| 
 \right]_{x_N} +
  \left| f_{\varepsilon} -f \right| \left[ \left| \Delta_i \psi \right|\right]_{x_N}\right)
  \end{align*}
  and, by using the inequality $\left||a_{\varepsilon}|^2 - |a|^2\right| \le (|a_{\varepsilon}| + |a|)|a_{\varepsilon}-a|$,
 \begin{align*}
 \left|
\left[ \left|\frac{\nabla_i \psi_{\varepsilon}}{\psi_{\varepsilon}} \right|^2 \right]_{x_N} f_{\varepsilon} 
- \left[ \left|\frac{\nabla_i \psi}{\psi} \right|^2 \right]_{x_N}f \right|
\le 
\mathrm{const.}& \big(
 \left[ \left|\nabla_i \psi_{\varepsilon} \right|^2 \right]_{x_N}
 |f_{\varepsilon} - f| 
 + \left[(|\nabla_i \psi_{\varepsilon}|  + |\nabla_{i}\psi|)  |\nabla_{i}\psi|
 |\psi_{\varepsilon} - \psi|\right]_{x_N}  \\
&+ [ (\left|\nabla_i \psi_{\varepsilon} - \nabla_i \psi \right| 
(|\nabla_i \psi_{\varepsilon}| + |\nabla_i \psi| )]_{x_N} 
\big) 
\end{align*}
on $K$. 
These bound and a similar calculation for the third term,
together with 
$f_\varepsilon \to f$ and $\psi_{\varepsilon} \to \psi$
 uniformly on any compact set and 
$\psi_{\varepsilon} \to \psi$ in $H^2(\mathbb{R}^{3N})$,
  lead to the desired result (\ref{1}).
 
 Therefore, we obtain
 \begin{align*}
\int_{\mathbb{R}^{3N}}  f(x) \Delta_i \phi(x) dx
&= \lim_{\varepsilon \to 0}\int_{\mathbb{R}^{3N}}  f_{\varepsilon}(x) \Delta_i \phi(x) dx \\
 &=  \lim_{\varepsilon \to 0}\int_{\mathbb{R}^{3N}}  \phi(x) \Delta_if_{\varepsilon}(x) dx
= \int_{\mathbb{R}^{3N}}  \phi(x) F_i dx
 \end{align*}
 for any test function $\phi \in C_c^{\infty}(\mathbb{R}^{3N})$.
Consequently, we obtain
$F_i = \Delta_i f \in L^1_{\mathrm{loc}}(\mathbb{R}^{3N})$
as the weak derivative.
The proof of $\nabla_i f \in L^2_{\mathrm{loc}}(\mathbb{R}^{3N})$,
as the weak derivative, is essentially the same.

Finally, we note that
 $\left| \nabla_N [\ln \psi_{\varepsilon}]_{x_N} \right|^2 = \left[ \left(\partial \psi_{\varepsilon}/{\partial r_N}  \right) /{\psi_{\varepsilon}} \right]_{x_N} ^2
 \le \left[ \left|\left( \nabla_N \psi_{\varepsilon}\right) / \psi_{\varepsilon} \right|^2\right]_{x_N} $ 
and
 $\left| \nabla_i [\ln \psi_{\varepsilon}]_{x_N} \right|^2 
\le \left[ \left|\left(\nabla_i \psi_{\varepsilon}\right) / \psi_{\varepsilon} \right|^2\right]_{x_N} $
for $i= 1, \dots, N-1$,
  by the Cauchy-Schwarz inequality. 
  Then (\ref{0}) implies that 
  $\Delta_i f_{\varepsilon} \le [ (\Delta_i \psi_{\varepsilon}) / \psi_{\varepsilon}]_{x_N} f_{\varepsilon}$.
  Since $\Delta_i f_{\varepsilon}$ converges to $F_i = \Delta_i f$, this yields
 \begin{align*}
 \int_{\mathbb{R}^{3N}} \phi \left[\frac{\Delta_i \psi}{\psi}\right]_{x_N} f dx
&=\lim_{\varepsilon \to 0}  \int_{\mathbb{R}^{3N}} \phi \left[\frac{\Delta_i \psi_{\varepsilon}}{\psi_{\varepsilon}}\right]_{x_N} f_{\varepsilon} dx \\
&\ge 
 \lim_{\varepsilon \to 0} \int_{\mathbb{R}^{3N}} \phi \Delta_i f_{\varepsilon} dx\
 = \int_{\mathbb{R}^{3N}} \phi \Delta_i f dx, \quad
\end{align*}
$i = 1, \dots, N$, for every non-negative test function $\phi \in C_c^{\infty}(\mathbb{R}^{3N})$. 
 \end{proof}
 
Using the decomposition
\[
H_N^Z = H_{N-1}^Z - \Delta_N - \sum_{j=1}^K z_j |x_N - R_j|^{-1} + \sum_{i=1}^{N-1}|x_i - x_N|^{-1},
\]
 we observe that
 \begin{align*}
 0 &= \left[  \frac{(H_N^Z - E_N^Z)\psi}{\psi} \right]_{x_N} f \\
 &= \sum_{i=1}^{N-1} \left(- \left[\frac{\Delta_i \psi}{\psi}\right]_{x_N}  -
  \sum_{j=1}^K z_j |x_i - R_j|^{-1}\right)f
 + \sum_{1\le i < j\le N-1} |x_i - x_j|^{-1}f  \\
 & - \left[\frac{\Delta_N \psi}{\psi}\right]_{x_N} f
  - \sum_{j=1}^K z_j \left[|x_N - R_j|^{-1} \right]_{x_N} f
 + \sum_{i=1}^{N-1} \left[|x_i - x_N|^{-1}\right]_{x_N} f  -E_{N-1}^Zf.
 \end{align*}
  Here we have used the assumption $E_N^Z = E_{N-1}^Z$.
  Now, for $|x_N| > R =  \max_{1 \le j \le K} |R_j|$, our hypothesis $Z= N-1$ leads to the bound
\[
-\sum_{j=1}^K z_j \left[|x_N - R_j|^{-1}\right]_{x_N}  +\sum_{i=1}^{N-1}\left[ |x_i - x_N|^{-1} \right]_{x_N}  
\le 0
\]
since $\left[|y - x_N|^{-1}\right]_{x_N}  = \min(|y|^{-1}, |x_N|^{-1})$ for $y \in \mathbb{R}^3$
by Newton's theorem (followed from an integration in polar coordinates~\cite[Theorem 5.2]{LS}).
This is the crucial place where averaging over $x_N$ helps.

We recall Lemma \ref{calculus} to conclude that
\[
0 \le (H_{N-1}^Z - E_{N-1}^Z -\Delta_N)f
\]
on $|x_N| > R$.
In addition, by Zhislin's theorem and the HVZ theorem,
there is a normalized eigenfunction $\phi$ of $H_{N-1}^Z$
corresponding to the eigenvalue $E_{N-1}^Z$ .
 As in the case of $\psi$, we can assume that $\phi > 0$, etc.
Thus, there is a sequence
 $\phi_n \in C_c^{\infty} (\mathbb{R}^{3(N-1)})$  such that
  $\phi_n \to \phi$ in $H^2(\mathbb{R}^{3(N-1)})$ with
 $\phi_n \ge 0$.
For any non-negative function $\eta \in C_c^{\infty} (\mathbb{R}^{3})$
with the support in the set $|x| > R$,
we define a non-negative, smooth, and compactly supported function by
 $g_n(x_1, \dots, x_N) = \phi_n(x_1, \dots, x_{N-1}) \eta(x_N)$.

 Lemma \ref{calculus} now implies that 
 \[
 0 \le \left\langle g_n, \left(H_{N-1}^Z - E_{N-1}^Z  - \Delta_N\right)f \right\rangle
  = \left\langle \left(H_{N-1}^Z - E_{N-1}^Z  - \Delta_N\right)g_n, f \right\rangle.
  \]
 As $n \to \infty$ the right side converges to $\langle \left(H_{N-1}^Z - E_{N-1}^Z  - \Delta_N\right)g, f \rangle$,
 where $g = \lim_{n \to \infty} g_n$. 
Since $H_{N-1}^Z \phi = E_{N-1}^Z \phi$, it follows that
  \[
  0 \le \left\langle-  \Delta_N g, f \right\rangle_{L^2(\mathbb{R}^{3N})} = \left\langle  - \Delta \eta,  v \right\rangle_{L^2(\mathbb{R}^3)}
  \]
  for all $0 \le \eta \in C_c^{\infty} (\mathbb{R}^3)$ with the support in the set $\{ x \, :\,  |x| > R\}$,
  where 
  \[
  v(r_N) = \int_{\mathbb{R}^{3(N-1)}}  \phi(x_1, \dots, x_{N-1})
   f(x_1, \dots, x_{N-1}, r_N) dx_1 \cdots dx_{N-1}.
  \]
Hence, we find that $v$ is superharmonic, that is,
  \[
  0 \le - \Delta v
  \]
  on $|x_N| > R$, in distributional sense.
Also, $v_{\varepsilon} = h_{\varepsilon}\ast v$ is superharmonic in 
$|x_N| > R + \varepsilon$
for sufficiently small $\varepsilon > 0$, 
where $h_{\varepsilon}$ is some mollifier in $\mathbb{R}^3$. 
By the positivity of continuous functions $f$ and $\phi$, we can choose a constant $C_R > 0$  so that
  $v_{\varepsilon}(x) \ge C_R |x|^{-1}$  for all  $x \in \mathbb{R}^3$ with $|x| = R + 1$.
  Since  $u(x) = C_R |x|^{-1}$ is harmonic in $|x| > 0$,
  we arrive at
  \[
  - \Delta (u - v_{\varepsilon}) \le 0 \quad \text{in }|x| > R + 1,
  \]
  and by the strong maximum principle~(\cite[Theorem 9.4]{LiLo}),
we infer that $u - v_{\varepsilon}$ takes 
  its maximum value on 
$\{ x \in \mathbb{R}^3 \, : \, |x| = R + 1\} \cup \{+ \infty \}$.
  Hence,
  $\lim_{|x| \to + \infty} ( u - v_{\varepsilon})(x) \le 0$ 
and $u \le v_{\varepsilon}$ on $|x| = R + 1$ imply that
  $u \le v_{\varepsilon}$ in $ |x| > R + 1$.
From this result, we have $u \le \lim_{\varepsilon' \to 0}v_{\varepsilon'} =v$
 for almost every $x$ with $|x| > R + 1$ and some subsequence 
$\varepsilon'$.
Since $C_R |x|^{-1} \le v(x)$ for almost all $|x| > R+1$, it immediately follows that
\[
\int_{\mathbb{R}^3}v(r_N)^2 dx_N  = + \infty.
\]
  From the fact that $f \le [\psi]_{x_N} $ (recall the first line of the proof of Lemma \ref{calculus})
  and the Cauchy-Schwarz inequality, we conclude
  \begin{align*}
  +\infty &= \int_{\mathbb{R}^3} v(r_N)^2dx_N  \\
  &= \int_{\mathbb{R}^3}
 \left( \int_{\mathbb{R}^{3(N-1)}} 
  \phi(x_1, \dots, x_{N-1}) f(x_1, \dots, x_{N-1}, r_N)dx_1 \cdots dx_{N-1}   \right)^2dx_N  \\
  &\le \int_{\mathbb{R}^{3N}} f(x_1, \dots, x_{N-1}, r_N)^2 dx_1 \cdots dx_N
  \le \int_{\mathbb{R}^{3N}}[\psi]^2 dx_1 \cdots dx_N
  \le \|\psi\|_2^2.
  \end{align*}
  This contradicts our assumption $\psi \in L^2(\mathbb{R}^{3N})$.


\end{document}